\documentclass[review]{elsarticle}

\usepackage{lineno,hyperref}

\usepackage{amsmath}
\usepackage{amsthm}
\usepackage{amssymb}
\usepackage{tabularx}
\usepackage{algorithm}
\usepackage{algorithmic}
\usepackage{graphicx}

\newtheorem{mydef}{Definition}

\newtheorem{lem}{Lemma}

\journal{Journal of \LaTeX\ Templates}

\begin{document}

\begin{frontmatter}

\title{Utility-efficient Differentially Private K-means Clustering based on Cluster Merging}

\author{Tianjiao Ni, Minghao Qiao, Zhili Chen\footnote{Zhili Chen is the corresponding author. Email: zlchen@ahu.edu.cn}, Shun Zhang, Hong Zhong}
\address{School of Computer Science and Technology, Anhui University, Hefei, China}
%
%
%

\begin{abstract}
Differential privacy is widely used in data analysis. State-of-the-art $k$-means clustering algorithms with differential privacy
typically add an equal amount of noise to centroids for each iterative computation. In this paper, we propose a novel differentially private $k$-means clustering algorithm, DP-KCCM, that significantly improves the utility of clustering by adding adaptive noise and merging clusters. Specifically, to obtain $k$ clusters with differential privacy, the algorithm first generates $n \times k$ initial centroids, adds adaptive noise for each iteration to get $n \times k$ clusters, and finally merges these clusters into $k$ ones. We theoretically prove the differential privacy of the proposed algorithm. Surprisingly, extensive experimental results show that: 1) cluster merging with equal amounts of noise improves the utility somewhat; 2) although adding adaptive noise only does not improve the utility, combining both cluster merging and adaptive noise further improves the utility significantly.
\end{abstract}


\end{frontmatter}


\section{Introduction}

With the rapid development of Internet technology, third-party applications have produced a large amount of user data. The correct use of these data is able to create incalculable value for governments, companies and individuals. How to extract useful information from user data is currently a hot research direction \cite{Chen2002Data,Tambe2016EFFECTIVE}. Clustering algorithms are widely used to complete this task in the field of data analysis \cite{Grabmeier2002Techniques,Patel2016The}. The goal of clustering is to divide elements of a dataset into different groups so that the elements in the same group have high similarity. There are a large number of clustering algorithms \cite{Sun2008Clustering,Antonenko2012Using,Ding2015Research,Mohebbi2016Iterative,Nissim2017Clustering,Stemmer2019Locally}. Among them, the $k$-means clustering is one of the most popular methods for numeric data.

Recently, many applications adopted the $k$-means clustering algorithm. Javadi et al. \cite{javadi2017classification} classified aquifer vulnerability using $k$-means cluster analysis. Han et al. \cite{han2018kclp} factored out $k$-means cluster-based location privacy protection scheme for Internet of Things. Shakeel et al. \cite{shakeel2018cloud} used $k$-means clustering to diagnosis of diabetes mellitus. Wu et al. \cite{wu2018deep} used $k$-means in the compressing convolutions of convolutional neural network. Reza et al. \cite{reza2019rice} adopted $k$-means clustering with graph-cut segmentation to estimate rice yield. Omrani et al. \cite{omrani2019land} presented the artificial neural network-based land transformation model, which uses the $k$-means clustering algorithm implemented within the Spark high-performance compute environment. However, most of these applications fail to consider disclosure of sensitive information, which might bring immeasurable threats to users \cite{xiong2014secure,gao2010approach}.

To solve the privacy problem, differential privacy \cite{Dwork2006Differential,Dwork2006Calibrating} is proposed as a powerful privacy protection technique and has been extensively used \cite{Dwork2008Differential,dankar2012application,Xiao2015Protecting}.
Recently, several state-of-the-art $k$-means clustering algorithms with differential privacy have been proposed. For example,  Yu et al. \cite{Yu2016Outlier} presented a differentially private $k$-means clustering scheme and improved its utility by selecting initial centroids with the distribution density of elements. Su et al. \cite{su2016differentially} analyzed several existing differentially private $k$-means clustering algorithms and improved one of them by selecting initial centroids based on the concept of sphere packing. However, these algorithms still suffer the issue of lacking high utility due to adding large amounts of noise. Thus, how to improve the utility of differentially private $k$-means clustering remains as a key question.

In this paper, to address the above-mentioned utility issue, we propose a novel differentially private $k$-means clustering algorithm based on cluster merging (DP-KCCM). DP-KCCM first partitions the data into $n \times k$ clusters with differential privacy, and then merges these clusters into required $k$ ones. The main idea is that the Laplace noises added to cluster centroids are random, and cluster merging would cancel the noises each other, and thus improve the utility. More interestingly, we find that combining cluster merging with adaptive noise is able to further improve the cluster utility.

The main contributions of this paper are as follows:
1) We propose a utility-efficient, differentially private $k$-means clustering algorithm based on cluster merging.
2) We design a privacy budget (i.e., the privacy parameter of differential privacy, cf. Definition~\ref{def:dp}) allocation to work with cluster merging to further improve the cluster utility.
3) Extensive experimental results show that our algorithm is superior to the state-of-the-art ones.

The rest of this paper is organized as follows. In Section~\ref{sec:background}, we introduce the background knowledge of this paper. Section~\ref{sec:dp-kccm} describes our algorithm in detail and establishes that the algorithm satisfies differential privacy. In Section~\ref{sec:evaluation}, we carry out extensive experiments, and compare the utility of our algorithm with those of the existing algorithms. We conclude in Section~\ref{sec:conclusion}.

\section{Background}\label{sec:background}

In this section, we first introduce the notion of differential privacy and the algorithms of K-means clustering, we then present the problem statement. Some notations used in this paper are described in Table~\ref{tab:notation}.

\begin{table}[!htb]
\caption{Notations and descriptions}\label{tab:notation}
\begin{center}
\begin{tabularx}{11cm}{lX}  
\hline                      
Notations   & Descriptions  \\
\hline
$D$   & dataset \\
$N$   & the number of data points in the dataset \\
$d$   & the dimension of the dataset \\
$x_i$   & $i$-th data point in the dataset and ranging from $[-r, r]^d$\\
$k$   & the number of the clusters \\
$C$   & the set of the centroids \\
$C_j$   & $j$-th centroid \\
$C_j^i$    & $i$-th dimension of $j$-th centroid\\
$c_j$   & $j$-th noisy centroid\\
$c_j^i$    & $i$-th dimension of $j$-th noisy centroid\\
$C_j^\ast$   & $j$-th cluster \\
$\epsilon$   & the privacy budget \\
$sum(\cdot)$   & the sum of data point values in a cluster \\
$sum^\prime(\cdot)$   & $sum^\prime(\cdot)$ is $sum(\cdot)$ with noise \\
$num(\cdot)$   & the number of data points in a cluster \\
$num^\prime(\cdot)$   & $num^\prime(\cdot)$ is $num(\cdot)$ with noise \\
$dist(x,y)$   & the distance between data points $x$ and $y$ \\
$\Delta$   & the global sensitivity of an iteration \\
$iter$   & the number of iterations \\
\hline
\end{tabularx}
\end{center}
\end{table}

\subsection{Differential Privacy}\label{sec:dp}
The notion of differential privacy requires that the outputs of a data analysis mechanism should be similar over any two adjacent datasets. The formal definition is as follows.
\begin{mydef}\label{def:dp}
({$\epsilon$}-Differential Privacy \cite{Dwork2006Differential}). A randomized algorithm $\mathcal{M}$ satisfies {$\epsilon$}-differential privacy ($\epsilon$-DP), if and only if for any pair of neighboring datasets $D$ and ${D^\prime}$, and any $\mathcal{S}$ {$\subseteq$} Range($\mathcal{M}$), we have
\begin{equation}\label{equ:dp}
Pr[\mathcal{M}(D)=\mathcal{S}] \leq e^\epsilon \cdot Pr[\mathcal{M}(D^\prime)=\mathcal{S}]
\end{equation}
\end{mydef}

In this definition, $D$'s neighboring dataset $D^\prime$ can be obtained by adding an element to or removing an element from $D$, and they can be denoted by $D \simeq D^\prime$. The Range($\mathcal{M}$) represents the set of all possible outputs of the algorithm $\mathcal{M}$. It is worth noting that the parameter $\epsilon$ is called \emph{privacy budget}, which indicates the privacy level. A smaller $\epsilon$ value
means more similar outputs resulted from neighboring datasets due to Eq.~\eqref{equ:dp}, and thus represents stronger privacy achieved. On the other hand, the greater the $\epsilon$ value is, the weaker the privacy is preserved. In the extreme, when there is no differential privacy protection, it is equivalent that the $\epsilon$ is infinitely large, and thus the privacy may be easily disclosed.

Differential privacy has good properties in composition of multiple algorithms, which are described in Lemmas~\ref{lem:parallel}, \ref{lem:sequential} and \ref{lem:post-processing}.
\begin{lem}[Parallel Composition \cite{Li2016Differential}]\label{lem:parallel}
If there are algorithms $\mathcal{M}_1,\cdots,\mathcal{M}_k$ satisfying $\epsilon_1,\cdots,\epsilon_k$ -DP, respectively, then for disjoint datasets $D_1,D_2,\cdots,D_k$, composition algorithm $\mathcal{M}(\mathcal{M}_1(D_1),\cdots,\mathcal{M}_k(D_k))$ provides
$(\max_{i \in \{1,..,k\}}\epsilon_i)$-DP.
\end{lem}
Lemma~\ref{lem:parallel} shows that if the input datasets are disjoint, the privacy level provided by a parallel composition depends on the algorithm with the lowest privacy level, namely the one with the largest privacy budget.

\begin{lem}[Sequential Composition \cite{dwork2014algorithmic,Li2016Differential}]\label{lem:sequential}
If there are algorithms $\mathcal{M}_1$ satisfying $\epsilon_1$-DP, and $\mathcal{M}_2$ satisfying $\epsilon_2$-DP, then $\mathcal{M}(D) = \mathcal{M}_1(\mathcal{M}_2(D),D)$ satisfies ($\epsilon_1 + \epsilon_2$)-DP.
\end{lem}
Lemma~\ref{lem:sequential} implies that if multiple algorithms are applied to the same data set sequentially, the resulted privacy budget is the sum of these algorithms' privacy budgets.

\begin{lem}[Post-processing \cite{dwork2014algorithmic}]\label{lem:post-processing}
If there is an algorithm $\mathcal{M}_1(\cdot)$ satisfying $\epsilon$-DP, then for any algorithm $\mathcal{M}_2(\cdot)$,
$\mathcal{M}_2(\mathcal{M}_1(\cdot))$ satisfies $\epsilon$-DP.
\end{lem}
Lemma~\ref{lem:post-processing} shows the post-processing property of differential privacy, namely, if an algorithm takes as input the output of another algorithm that satisfies $\epsilon$-differential privacy, then the resulted algorithm still satisfies $\epsilon$-differential privacy.

In this paper, we use the Laplace mechanism \cite{Dwork2006Calibrating} to design algorithms. The Laplace mechanism preserves differential privacy by adding random noise satisfying the Laplace distribution to any query function $f$ (e.g., count query) over $D$. The magnitude of noise depends on the sensitivity of $f$, $\Delta f$, which represents the maximum deviation of the query result on any adjacent datasets. For instance, $\Delta f=1$ for count query. The idea of Laplace mechanism is to add smallest but sufficient noise to any query result of $f$, such that the query results of any neighboring datasets (only differing in an element) are indistinguishable, and thus the personal privacy (i.e., any single element information) cannot be inferred from query results. The \emph{Laplace mechanism} $\mathcal{M}_f$ is given in Definition~\ref{def:laplace}, and its differential privacy is guaranteed by Lemma~\ref{lem:laplace}

\begin{mydef}\label{def:laplace}
(Laplace Mechanism \cite{Dwork2006Calibrating}) The Laplace mechanism is defined as:
\begin{equation}
\mathcal{M}_L(D,f,\epsilon) = f(D) + (Y_1, Y_2,\cdots, Y_d)
\end{equation}
where $f(D)$ is a given query function $f(D):D \rightarrow R^d$ with sensitivity
\begin{equation}
\Delta f = \max_{(D,D^\prime):D \simeq D^\prime} \|f(D) - f(D^\prime)\|_1
\end{equation}
and $Y_i (1 \le i \le d)$ are i.i.d. random variables drawn from $Lap(\frac{\Delta f}{\epsilon})$.
\end{mydef}
The probability density function of Laplace distribution is as follows:
\begin{equation}
Lap(b) = Lap(x|b) = \frac{1}{2b}e^{-|x|/b}
\end{equation}
where for the Laplace mechanism, $b = \Delta f/\epsilon.$

\begin{lem}\cite{dwork2014algorithmic,Li2016Differential}\label{lem:laplace}
The Laplace mechanism preserves $\epsilon$-DP.
\end{lem}

\subsection{K-means Clustering Algorithms}
The $k$-means clustering divides a set of data points into different subsets. In the clustering, there are typically two following steps.
\subsubsection{Initial Centroid Selection}
Given a value $a > 0$, paper \cite{su2016differentially} randomly selects $k$ initial centroids one by one. Each selection of a centroid must follow two principles: 1) the distance between any centroid and the boundary of the domain is at least $a$. 2) the distance between any two centroids is at least $2a$. If a randomly selected centroid does not satisfy the above two conditions, it is discarded and another centroid is reselected until $k$ initial centroids are obtained. When it fails to get $k$ centroids several times, it may be that the given value $a$ is too large, and then a smaller $a$ is tried again. The initial value of $a$ can be determined according to the domain size of data points. In our context, since the data is normalized into $[-1,1]$, the initial value of $a$ is set to 0.5. During the experiment, the optimal value of $a$ can be obtained by the binary search or it is simply set by experience. The above process of selecting the initial centroids only depends on the domain of data points rather than the data points themselves, so this process does not impact the privacy, and it can be applied directly in the privacy-preserving algorithms.

\subsubsection{K-means Clustering}
For a dataset $D = \{x_1, x_2, \cdots, x_N\}$, $x_i \in R^d$, the standard $k$-means clustering aims at partitioning data into $k$ disjoint subsets $(C_1^\ast, C_2^\ast, \cdots, C_k^\ast)$. The evaluation metric of clustering results called the Normalized Intra-Cluster Variance (NICV)\cite{su2016differentially} is as follows.
\begin{equation}
\frac{1}{N}\sum_{j=1}^k \sum_{x_i \in C_j^*} \|x_i - C_j\|_2 
\end{equation}
where $C_j$ is the centroid of the cluster $C_j^*$, and the smaller the NICV value, the better the clustering result.

Specifically, the algorithm selects $k$ data points as the initial centroids by initial centroids selection algorithm, then the quality of the centroids is improved iteratively until the centroids do not change. In each iteration, the algorithm traverses all the data points of the dataset and assigns the data points to the nearest cluster, then updates the centroid of each cluster.
\begin{equation}
C_j^t = \frac{\begin{matrix} \sum_{x_i \in C_j^*} x_i^t \end{matrix}}{|C_j^*|},\forall t \in \{1,..,d\} 
\end{equation}
where $C_j^t$ is the $t$-th dimension of $j$-th centroid and $x_i^t$ is the $t$-th dimension of $x_i$.

\subsection{Problem Statement}

We focus on the privacy problem for data analysis as follows. A trustable data holder (e.g., a government agency) collects personal records from a great number of users, analyzes these data with certain machine learning algorithms (e.g., in our context K-means clustering algorithm), and publishes the analytic results to dishonest third parties, which make use of these results for some statistical purpose. However, the privacy problem is that the dishonest third parties may also infer the personal privacy through exploiting the analytic results, beyond the legitimate use, if there is no appropriate privacy protection measures, as shown in Figure~\ref{fig:problem}. This results in the disclosure of individual privacy.

\begin{figure}[htb]

\centering
\includegraphics[width=0.8\textwidth]{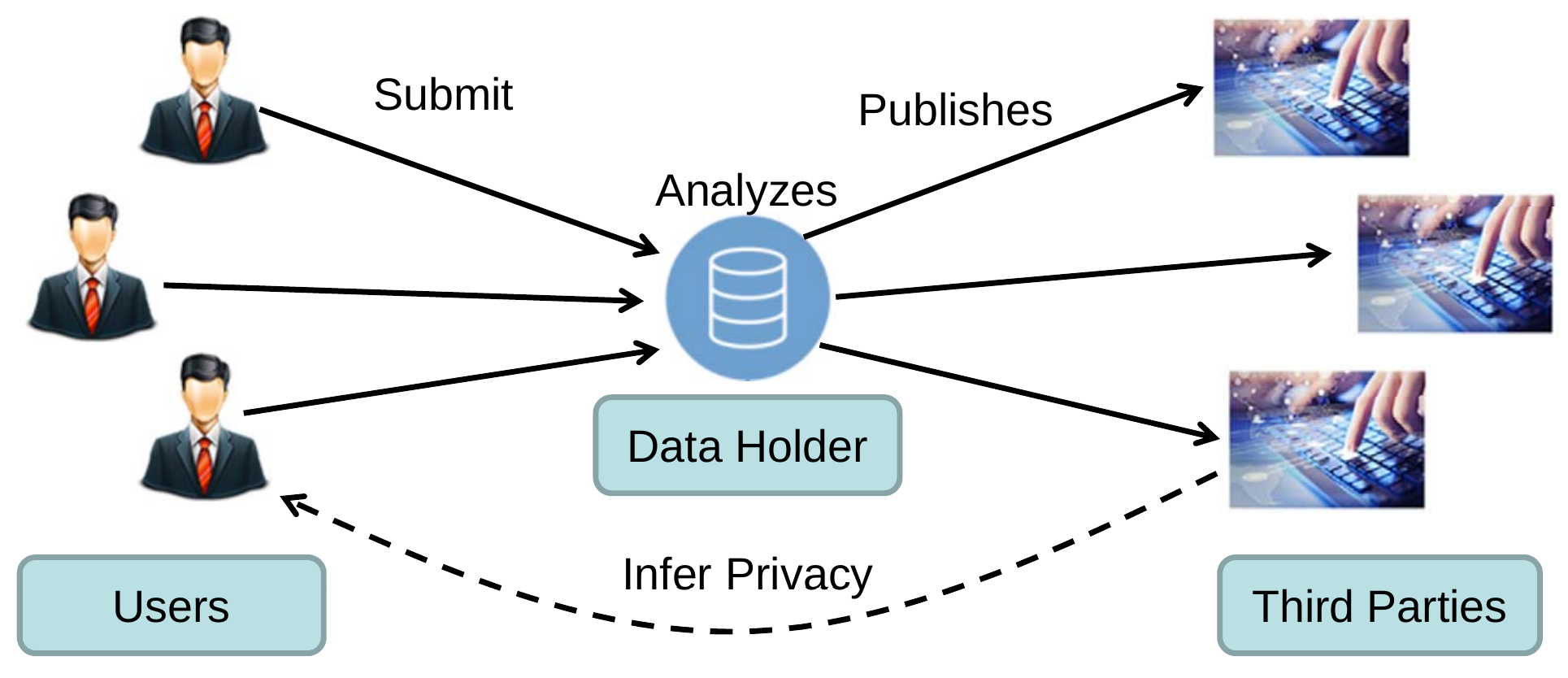}
\caption{Problem Model}\label{fig:problem}
\end{figure}

We use the notion of differential privacy to solve the privacy problem described above. As described in Section~\ref{sec:dp}, differential privacy ensures that the analytic results of any two neighboring datasets are similar by adding appropriate noise to the results. This means that the differentially private analytic result of a dataset remains roughly the same when any personal record opts into or out of the dataset. Conversely, the inference of any personal record (i.e., individual privacy) from the differentially private analytic results is thus hard.

In our context, we design a K-means clustering algorithm with $\epsilon$-differential privacy to protect individual privacy (i.e., any data point) from being inferred from the clustering results by dishonest third parties. A proper $\epsilon$ value can be chosen to ensure a certain level of privacy, e.g., $\epsilon=1$ (Note that the smaller $\epsilon$, the stronger the privacy). Furthermore, since $\epsilon$-differential privacy requires that noise be added to the clustering process, we aim to add as little noise as possible while preserving $\epsilon$-differential privacy, to improve the clustering utility.

\section{DP-KCCM Algorithm}\label{sec:dp-kccm}
In this section, we propose our differentially private $k$-means clustering algorithm (DP-KCCM) and prove its privacy in detail.

Considering previous works on $k$-means clustering with differential privacy, such as DPLloyd \cite{Blum2005Practical,McSherry2010Privacy} and DPLloyd-Impr \cite{su2016differentially}, they all added equal amounts of noise to the centroids of each iteration in the clustering process. Moreover, it is implicitly suggested that privacy budget should be divided equally across the iterations of the clustering \cite{su2016differentially}. It seems hard to improve the cluster utility merely through the privacy budget allocation. Our idea is that we may reduce the amounts of noise added by merging some noisy clusters and canceling the amounts of noise. Additionally, we may combine cluster merging with privacy budget allocation to further improve the utility.

Therefore, we are going into two questions: 1) Can we add noise adaptively in the process of iteration to improve the utility? 2) Can we merge adjacent clusters to reduce the noise added, and hence improve the utility? We describe these two aspects in detail.

\subsection{Idea 1}
In our differentially private $k$-means clustering algorithm, initial centroids are first selected using the initial centroids selection algorithm \cite{su2016differentially} (cf. Section 2.2.1), which ensures that the initial centroids are separated as much as possible. Then, these centroids are iteratively updated. Intuitively, in the first several iterations of the algorithm, the centroids change greatly, and we could inject relatively more noise. As the number of iterations increases, the changes of cluster centroids become less, and we could add a small amount of noise to ensure better clustering results. Since the noise volume is controlled by $\epsilon$, we introduce the partition of $\epsilon$ as follows.

According to the previous analysis, we know that as the number of iterations increases, the clustering centroids tends to be stable, and hence the noise added should become smaller, and $\epsilon$ should become bigger. Thus, our privacy budget allocation policy is to increase $\epsilon$ share gradually as the cluster centroids are updated iteratively.

We do a lot of experiments based on this allocation policy, and find that the increase of the $\epsilon$ share should be relatively slow. Otherwise, the first several $\epsilon$ shares would become too small that the cluster result deteriorates severely. Finally we settled on the following division. The clustering algorithm in this paper has carried out 12 iterations, and we set $\epsilon$ share for each iteration as follows: the values of $\epsilon$ shares for 1st to 4th iterations are $\frac{1}{24}\epsilon$, those for 5th to 8th iterations are $\frac{1}{12}\epsilon$, and those for 9th to 12th iterations are $\frac{1}{8}\epsilon$.

\subsection{Idea 2}\label{sec:idea2}
We know that the noise added to each centroid is random. Can we reduce the influence of noise on the centroid by merging adjacent clusters? For $k$-means clustering algorithm, we can divide dataset into $n \times k$ clusters. After clustering, we merge $n \times k$ clusters into $k$ ones. By merging multiple clusters, the noises added to clusters were empirically proved to cancel each other out.

We describe this idea with the example as shown in Figure~\ref{fig:example}. Suppose there are four points $A$, $B$, $C$ and $D$, which are probably clustered into the same category in a differentially private clustering. There are two ways to do this. The first way is to cluster these points into a category one shot without cluster merging, and get its noisy clustering centroid $C^{\star}_{ABCD}$. The second way is to first cluster the points into two clusters, with $A$, $B$ in one and $C$, $D$ in the other, compute their respective cluster centroids $C^{\star}_{AB}$ and $D^{\star}_{CD}$, and then merge the two clusters to get the final cluster centroid $C^{\star\star}_{ABCD}$. As long as the noises added to centroids $C^{\star}_{AB}$, $C^{\star}_{CD}$ and $C^{\star}_{ABCD}$ are roughly the same, namely the distances from the noiseless centroids $C_{AB}$, $C_{CD}$ and $C_{ABCD}$ to the noisy counterparts are approximately equal, the merged cluster centroid $C^{\star\star}_{ABCD}$ is probably less noisy the the centroid $C^{\star}_{ABCD}$, and thus this likely leads to a better clustering utility.

\begin{figure}[htb]

\centering
\includegraphics[width=0.6\textwidth]{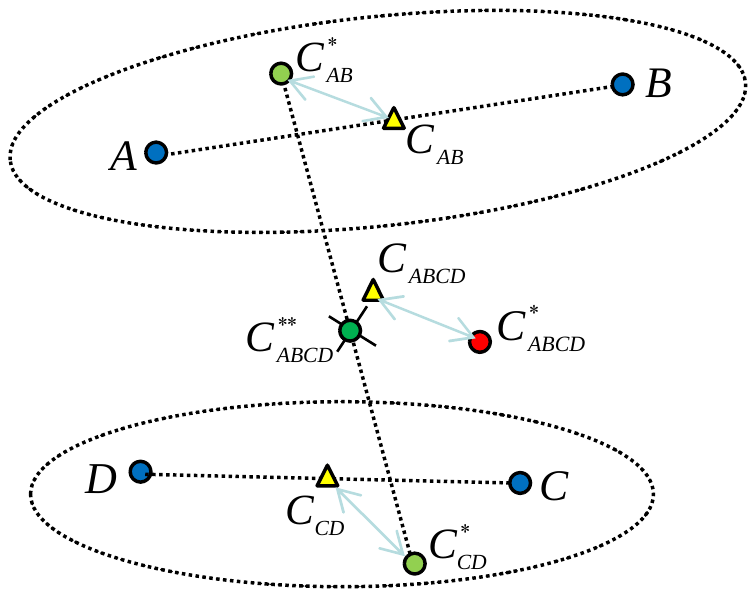}
\caption{An example for cluster merging.}\label{fig:example}
\end{figure}


\subsection{Algorithm DP-KCCM}
We combine idea 1 with idea 2 to design the algorithm DP-KCCM, which can be described with the following steps.
\begin{enumerate}[(1)]
\item Obtain $n \times k$ initial centroids by initial centroids selection algorithm.
\item Divide data points into $n \times k$ clusters.
\item Recalculate the centroids.
\item Add adaptive noise according to the number of iterations (cf. Section 3.1).
\item Repeat Steps (2), (3) and (4) until the maximum number of iterations is reached.
\item Merge $n \times k$ clusters into $k$ clusters.
\end{enumerate}
The formal description of the algorithm is shown in Algorithm~\ref{alg:dp-kccm}. We also detail the key steps in the following.

\renewcommand{\algorithmicrequire}{\textbf{Input:}}  
\renewcommand{\algorithmicensure}{\textbf{Output:}}  
\begin{algorithm}[!htb]
    \caption{DP-KCCM}\label{alg:dp-kccm}
    \label{alg:1}
    \begin{algorithmic}
    \REQUIRE {dataset $D$, cluster number $k$, clustering number $max\_clustering = 12$, global sensitivity $\Delta = d \cdot r + 1$, privacy budget $\epsilon = \sum_{i=1}^{max\_clustering} \epsilon_i$.}
    \ENSURE {The $k$ centroids.}
    \STATE initialize $C = n \times k$ centroids by initial centroid selection\

    \FOR{$iter \leftarrow 1$ to $max\_clustering$}
        \STATE Get $n \times k$ clusters through the standard $k$-means algorithm
        \STATE Recalculate the centroid of each cluster
        \FOR{$j \leftarrow 1$ to $n \times k$}
            \FOR{$i \leftarrow 1$ to $d$}
                \STATE $sum^{\prime}(C^*_j)[i] = sum(C^*_j)[i] + Lap(\frac{\Delta}{\epsilon_{iter}})$
            \ENDFOR
            \STATE $num^\prime(C^*_j) = num(C^*_j) + Lap(\frac{\Delta}{\epsilon_{iter}})$
            \STATE $c_j = \frac{sum^\prime(C^*_j)}{num^\prime(C^*_j)}$
        \ENDFOR
    \ENDFOR

    \STATE $C\_num = |C|$
    \WHILE{$C\_num > k$}
        \STATE Find the two nearest cluster $C^*_p$, $C^*_q$ and combine them into cluster $C^*_o$:
        \FOR{$i \leftarrow 1$ to $d$}
            \STATE $c^i_o = \min(c^i_p,c^i_q) + |c^i_p - c^i_q| \cdot \frac{num'(C^*_{m})}{num'(C^*_p) + num'(C^*_q)}$
        \ENDFOR
        \STATE $C\_num = C\_num -1$
    \ENDWHILE

    \STATE Get $k$ centroids

    \end{algorithmic}
\end{algorithm}

\subsubsection{Divide Data Points into $n \times k$ Clusters}
For a data point, we set two variables $min\_dist$ and $cent$. The former represents the distance from this data point to its centroid, and the latter is the index of the cluster (i.e., $cent \in [1, n \times k]$) containing the data point. Firstly, we traverse all the data points in the dataset and calculate the distances of each data point to all centroids, then each data point is assigned to the cluster determined by the corresponding minimum distance. Furthermore, we save the minimum distance value of each data point into $min\_dist$ and record the index of the corresponding centroid into $cent$. The distance from $x_i$ to the $j$-th centroid $C_j$ is computed as follow:
\begin{equation}
dist(x_i,C_j) = \|x_i - C_j\|_2. 
\end{equation}
where each dimension of data points $x_i$ is normalized to $[-r, r]$, and we choose $r=1$ in our context for simplicity.

\subsubsection{Add Noise to Centroids}
When all data points are divided, we obtain $n \times k$ clusters. Then we recalculate the centroid of each cluster to get $n \times k$ new centroids. The calculation of the centroid is as follows:
\begin{equation}
C_j = \frac{sum(C_j^*)}{num(C_j^*)}, \forall j \in \{1,\cdots,n \times k\} 
\end{equation}
where $sum(C_j^*) = \begin{matrix} \sum_{x_i \in C_j^*} x_i \end{matrix}, num(C_j^*) = |C_j^*|$.

In order to protect the information of data points, we need to add a certain amount of Laplace noise to both $d$-dimension sum of data points and the number of data points during the above calculation. And then we get $n \times k$ noisy centroids after noise addition.
The function for calculating noisy centroid $j$ is
\begin{equation}
c_j = \frac{sum'(C_j^*)}{num'(C_j^*)}. 
\end{equation}
with
\begin{equation}
sum^\prime(C^*_j) = sum(C^*_j) + (Y_1,\cdots,Y_d) 
\end{equation}
\begin{equation}
num^\prime(C^*_j) = num(C^*_j) + Y_{d+1} 
\end{equation}
where $Y_i$ (for $1 \le i \le d+1$) are i.i.d. random variables drawn from $Lap(\Delta/\epsilon_{iter})$, $\Delta = d \cdot r + 1$, $r$ is the maximum absolute value of each dimension, and $\epsilon_{iter}$ is the privacy budget for the current iteration $iter$.

Note that the global sensitivity $\Delta$ can be computed as follows. For each iteration, each data point is involved in answering $d$ sum queries and one count query. Moreover, each dimension of data points is
normalized to $[-r, r]$. Thus, the global sensitivity $\Delta = d \cdot r + 1$.

\subsubsection{Merge $n \times k$ Clusters into $k$ Clusters}
After the number of iterations reaches the specified value, we can obtain $n \times k$ clusters. And then we merge the two nearest clusters iteratively. We merge two clusters using their noisy centroids and noisy cluster sizes. Note that, to achieve differential privacy, cluster sizes are needed to add Laplace noise, and the noisy sizes does not represent real counts of elements in clusters, but are only used for computing merged centroids. After merging of cluster $C^*_p$ and cluster $C^*_q$, the $i$-th dimension of the centroid of the new cluster $C^*_o$ is
\begin{equation}
c^i_o = \min(c^i_p,c^i_q) + |c^i_p - c^i_q| \cdot \frac{num'(C^*_{m})}{num'(C^*_p) + num'(C^*_q)} 
\end{equation}
where $m = \arg \max(c^i_p,c^i_q)$, namely, $m = p$ if $c^i_p \ge c^i_q$, and $m = q$ otherwise.

\subsection{Privacy Analysis}
Theorem~\ref{theorem1} states that the DP-KCCM algorithm achieves $\epsilon$-differential privacy. We mainly prove the theorem below.
\newtheorem{theorem}{Theorem}
\begin{theorem}\label{theorem1}
The DP-KCCM algorithm preserves $\epsilon$-differential privacy.
\end{theorem}
\begin{proof}
We prove the theorem in the following three parts.

First, the initial centroid selection is independent of the data, so the privacy is not impacted in this step.

Second, we show that the algorithm achieves $\epsilon_{iter}$-differential privacy for iteration $iter$, and achieves $\epsilon$-differential privacy for all iterations due to the sequential composition property (cf. Lemma~\ref{lem:sequential}).

For iteration $iter$, each data point is involved in $d$ sum queries and one count query. Conversely, each iteration queries the function $f : D^d \to D^d \times N$ over each cluster, with the global sensitivity $\Delta f = \Delta = d \cdot r + 1$. Specifically, let $D$ and $D'=D-\{x\}$, for any point $x \in D$, be neighboring datasets. Let $D$ be divided into disjoint clusters $C_1^*, C_2^*, \cdots, C_{n \cdot k}^*$. Then, the neighboring dataset $D^\prime$ is correspondingly divided into disjoint clusters ${C_1^*}^\prime, {C_2^*}^\prime, \cdots, {C_{n \cdot k}^*}^\prime$, satisfying that there exists an $J$, such that ${C_J^*}^\prime = {C_J^*}-\{x\}$, and ${C_j^*}^\prime = C_j^*$ for $j \neq J$. Therefore, each iteration can be regarded as the parallel composition of mechanisms querying function $f(.)$ over $n \times k$ disjoint clusters (cf. Lemma~\ref{lem:parallel}), and differential privacy achievement is determined by the mechanism over cluster $C_J^*$ (since mechanisms over other clusters $C_j^*$ achieve $0$-differential privacy for $C_j^*={C_j^*}^\prime$).

Let $p(.)$ and $p'(.)$ denote the probability density functions of the mechanism over cluster $C_J^*$. For any point $v \in D^d \times N$, the probability density ratio between the cases of $C_J^*$ and ${C_J^*}^\prime$ are as follows.
\begin{align*}\small
\frac{p(v)}{p^\prime(v)} &= \frac{exp(-\frac{\epsilon_{iter}||f(C_J^*)-v||_1}{\Delta f})}{exp(-\frac{\epsilon_{iter}||f({C_J^*}^\prime)-v||_1}{\Delta f})}\\
                      &= exp(\frac{\epsilon_{iter}(||f({C_J^*}^\prime)-v||_1-||f(C_J^*)-v||_1)}{\Delta f})\\
                      &\leq exp(\frac{\epsilon_{iter} \cdot ||f(C_J^*)-f({C_J^*}^\prime)||_1}{\Delta f})\\
                      &\leq exp(\epsilon_{iter})
\end{align*}

Symmetrically, we have $\frac{p(v)}{p^\prime(v)} \ge exp(-\epsilon_{iter})$. The mechanism over $C_J^*$ achieves $\epsilon_{iter}$-differential privacy, and thus the iteration mechanism $\mathcal{DP-ITER}$ achieves $\epsilon_{iter}$-differential privacy according to the parallel composition property (cf. Lemma~\ref{lem:parallel}).

Since the computation of noisy centroids $c_j$ is a post-process of iteration mechanism $\mathcal{DP-ITER}$, iteration $iter$ achieves $\epsilon_{iter}$-differential privacy due to the post-processing property (cf. Lemma~\ref{lem:post-processing}).

For all iterations, the sequential composition property (cf. Lemma~\ref{lem:sequential}) is applied, and the resulted mechanism satisfies the $\epsilon$-differential privacy, where $\epsilon = \epsilon_1 + \epsilon_2 + \cdots + \epsilon_{max\_clustering}$.

Finally, merging $n \times k$ clusters into $k$ ones involves only the noisy cluster centroids and noisy cluster sizes, and it is actually a post-process of the composition of iterations, impacting nothing on the differential privacy achieved.

Therefore, we conclude that the DP-KCCM algorithm preserves $\epsilon$-differential privacy.
\end{proof}

\section{Performance Evaluation}\label{sec:evaluation}

\subsection{Methodology}

We implement the proposed differentially private $k$-means clustering algorithm, DP-KCCM, and do experiments to evaluate the algorithm based on six datasets. The detailed description of these six datasets is as Table~\ref{tab:datasets}. The data attributes contained in these six datasets are all of numerical type. We normalize the domain of each attribute to the range of $[-1, 1]$.

\begin{table*}[!htb]
\caption{\small Description of the Datasets.}\label{tab:datasets}
\begin{center}\small
\begin{tabular}	{|p{2cm}<{\centering}|c|c|c|p{5cm}<{\centering}|}
	\hline
	Dataset & tuples & dims & cluster & description \\ \hline
    \hline
	Blood & 748 & 5 & 4 & \multicolumn{1}{|m{5cm}|}{The dataset records individual blood donations, and is taken from the Blood Transfusion Service Center.} \\
    \hline
    Adult &  32561 & 6 & 5 & \multicolumn{1}{|m{5cm}|}{This is a census dataset that records personal information.} \\
    \hline
    Tripadvisor-review & 980 & 10 & 4 & \multicolumn{1}{|m{5cm}|}{The dataset is the reviews on destinations in 10 categories mentioned across East Asia.} \\
    \hline
    Electrical & 10000 & 13 & 5 & \multicolumn{1}{|m{5cm}|}{The dataset is the simulated data for the local stability analysis of the 4-node star system (electricity producer is in the center) implementing Decentral Smart Grid Control concept.} \\
    \hline
    Review-ratings & 5454 & 24 & 4 & \multicolumn{1}{|m{5cm}|}{The dataset contains google reviews on attractions from 24 categories across Europe.} \\
    \hline
    Credit-card & 30000 & 24 & 5 & \multicolumn{1}{|m{5cm}|}{The dataset contains customer default payments in Taiwan.} \\
    \hline
\end{tabular}
    \\
    \small All datasets are downloaded from website http://archive.ics.uci.edu/ml/datasets.php

\end{center}
\end{table*}

We mainly focus on comparisons of algorithm performances from the following two aspects:
\begin{itemize}
\item Comparing the effect of different algorithms with a fix $k$ value under different $\epsilon$ values.
\item Comparing the effect of different algorithms with a fix $\epsilon$ value under different $k$ values.
\end{itemize}

We compare the following four differentially private $k$-means clustering algorithms.
\begin{itemize}
\item $average\_k$ : $k$ initial centroids are generated by the initial centroid selection algorithm, and the average noise is added to all centroids during each iteration.
\item $allocation\_k$ : $k$ initial centroids are generated by the initial centroid selection algorithm, and adaptive noise is added to all centroids during each iteration.
\item $average\_nk$ : $n \times k$ initial centroids are generated by the initial centroid selection algorithm, and average noise is added to all centroids in the process of each iteration. When the clustering is stable, $n \times k$ clusters are combined into $k$ clusters.
\item $allocation\_nk$ : $n \times k$ initial centroids are generated by the initial centroid selection algorithm, and adaptive noise is added to all centroids during each iteration. When the clustering is stable, $n \times k$ clusters are combined into $k$ clusters.
\end{itemize}

The $average\_k$ algorithm is in fact the state-of-the-art algorithm, DPLloyd-Impr, which is reported to be the best one overall among several existing differentially private $k$-means clustering algorithms \cite{su2016differentially}. Thus, in our experiments, we use it a benchmark algorithm for the performance comparisons. Then, on the basis of $average\_k$ algorithm, the two ideas are introduced separately, generating algorithms $allocation\_k$ and $average\_nk$, respectively. Finally, the two ideas are combined into the algorithm $allocation\_nk$, which is the proposed DP-KCCM algorithm. We compare these four algorithms in the experiments to demonstrate the effectiveness of the two ideas.

The above algorithms all output $k$ centroids $\boldsymbol{C} = \{C_1, C_2, \cdots, C_k\}$. The quality of clustering are assessed by Normalized Intra-Cluster Variance (NICV). For all algorithms, we apply the initial centroid selection algorithm to get the initial centroids. In our experiments, we first adopt the initial centroid selection algorithm to generate 20 sets of initial centroids. Then, we run 50 times on each set of initial centroids. So we take the average of NICV in 1000 experiments. Through a lot of experiments, we make the following settings for some parameters in the process of cluster with each dataset. We found that the clustering tends to be stable after the number of iterations reaches 10, so we set $max\_clustering=12$. The experimental results obtained by merging $n \times k$ clustering into $k$ clusters when $n$ is set to 3 are relatively ideal.

Note that in this paper, we measure the clustering utility with NICV through extensive
experiments. The underlying reason is that NICV is the objective function. NICV is directly effected by the differentially private clustering. To some extent, we can regard NICV value as the utility of the differentially private algorithm, and NICV can reflect straightforwardly how the noise addition for achieving differential privacy impacts the clustering result. By using multiple datasets and averaging a great
number of independently random runs, we expect NICV to measure the clustering
utility reasonably. In the future, other reasonable measures can be investigated
for fully evaluating the clustering utility.

\subsection{Experimental Results.}

We now show two groups of experimental results, and make corresponding discussions.

\textbf{(1) Performance Comparison in term of $\epsilon$}

Figures \ref{fig:Blood-e} to \ref{fig:Credit-e} show the influence of different $\epsilon$ values on the clustering results for different datasets. The $k$ values are fixed at $4$ or $5$.
The six figures can be divided into three groups: Figures \ref{fig:Blood-e} and \ref{fig:Adult-e}, Figures \ref{fig:Tripadvisor-e} and \ref{fig:Electrical-e}, Figures \ref{fig:Travel-e} and \ref{fig:Credit-e}. The dataset dimensionality values for the first group are in the range $[1..10]$, those of the second group are in the range $[11..20]$, and those of the third group are in the range $[21..30]$.

From these 6 figures, we make the following observations. 1) For all cases, $allocation\_k$ performs worse than $average\_k$, which illustrates that the average allocation of privacy budget seems to be the best choice for differentially private $k$-means clustering algorithms without cluster merging. 2) Nearly for all cases, $average\_nk$ performs better than $average\_k$, which indicates that cluster merging indeed improves the clustering utility. 3) For all cases, $allocation\_nk$ performs significantly better than $average\_nk$, demonstrating the surprising result that combining both cluster merging and adaptive privacy budget allocation is able to further improve the clustering utility. 4) The performances of all algorithms become better in a similar way as the $\epsilon$ value increases, which shows the stability of the improvement over clustering utility. 5) It seems that the dimensionality values have no obvious impact on the utility improvement, and some datasets seems have more effect than others of utility improvement maybe due to their data characteristics.

\begin{figure*}[!htb]
\begin{minipage}[b]{0.5\linewidth}
\centering
\includegraphics[width=1\textwidth]{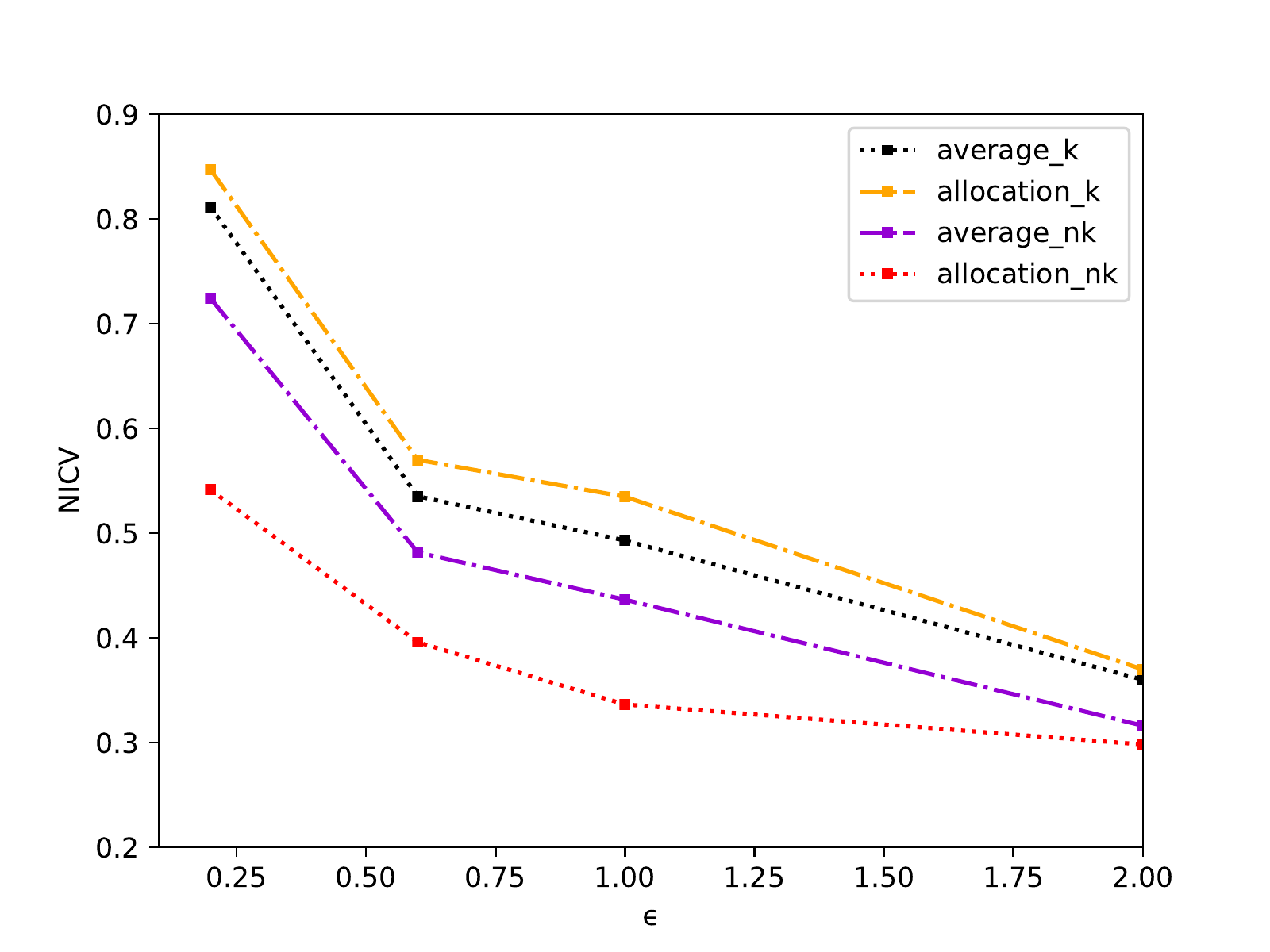}
\caption{Blood ($k=4$)}\label{fig:Blood-e}
\end{minipage}
\begin{minipage}[b]{0.5\linewidth}
\centering
\includegraphics[width=1\textwidth]{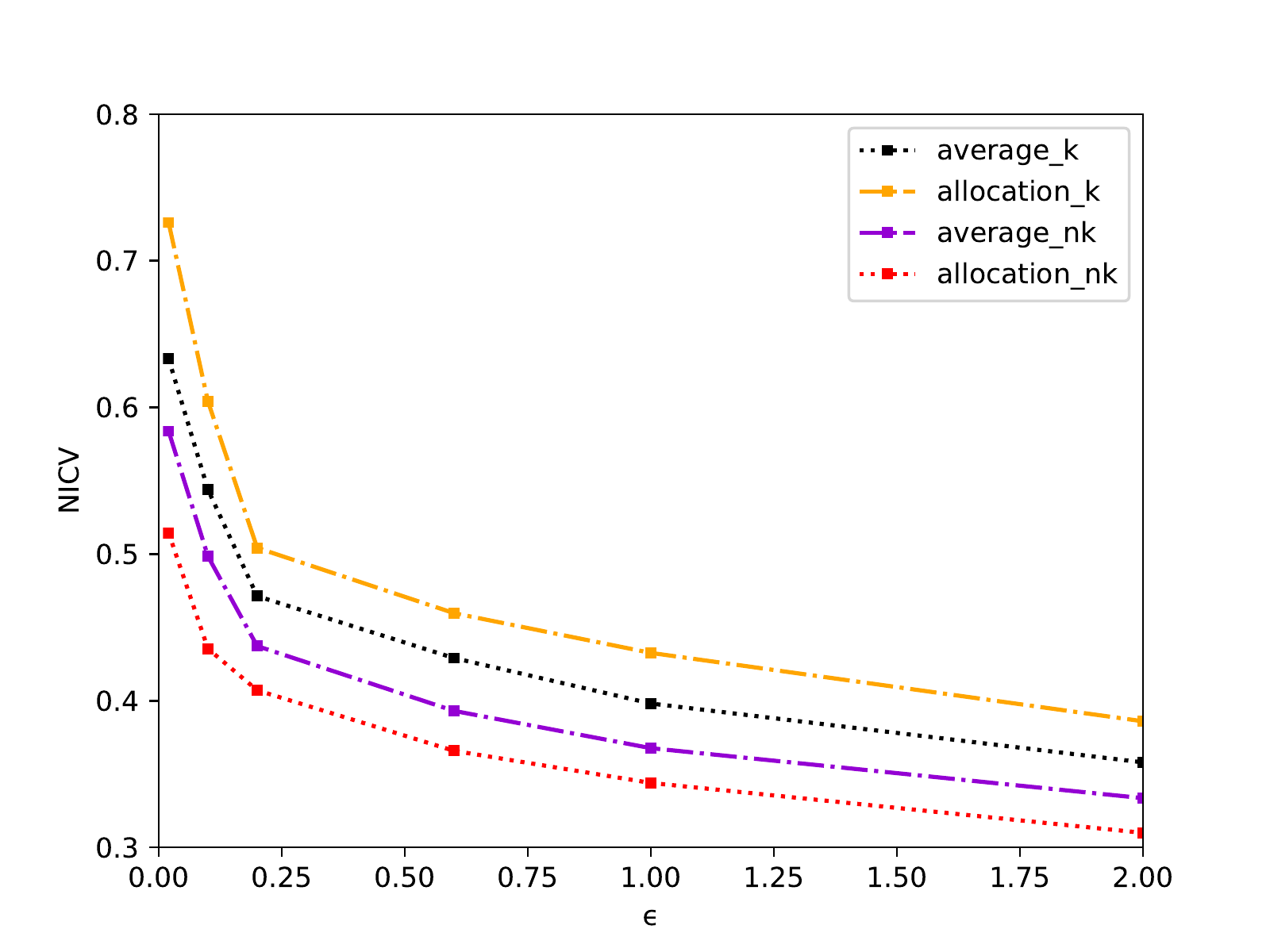}
\caption{Tripadvisor-review ($k=4$)}\label{fig:Tripadvisor-e}
\end{minipage}

\begin{minipage}[b]{0.5\linewidth}
\centering
\includegraphics[width=1\textwidth]{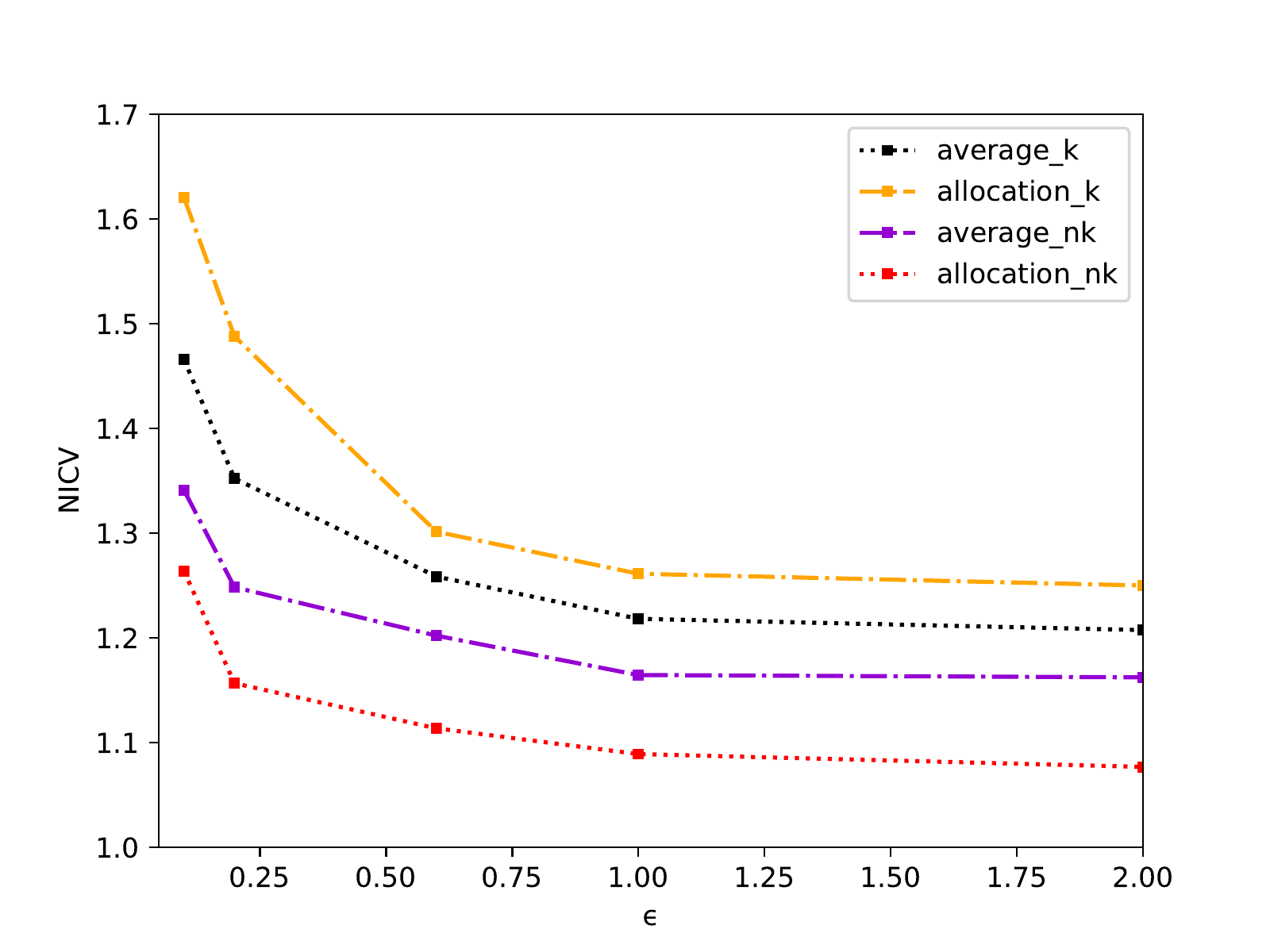}
\caption{Travel-review ($k=4$)}\label{fig:Travel-e}
\end{minipage}
\begin{minipage}[b]{0.5\linewidth}
\centering
\includegraphics[width=1\textwidth]{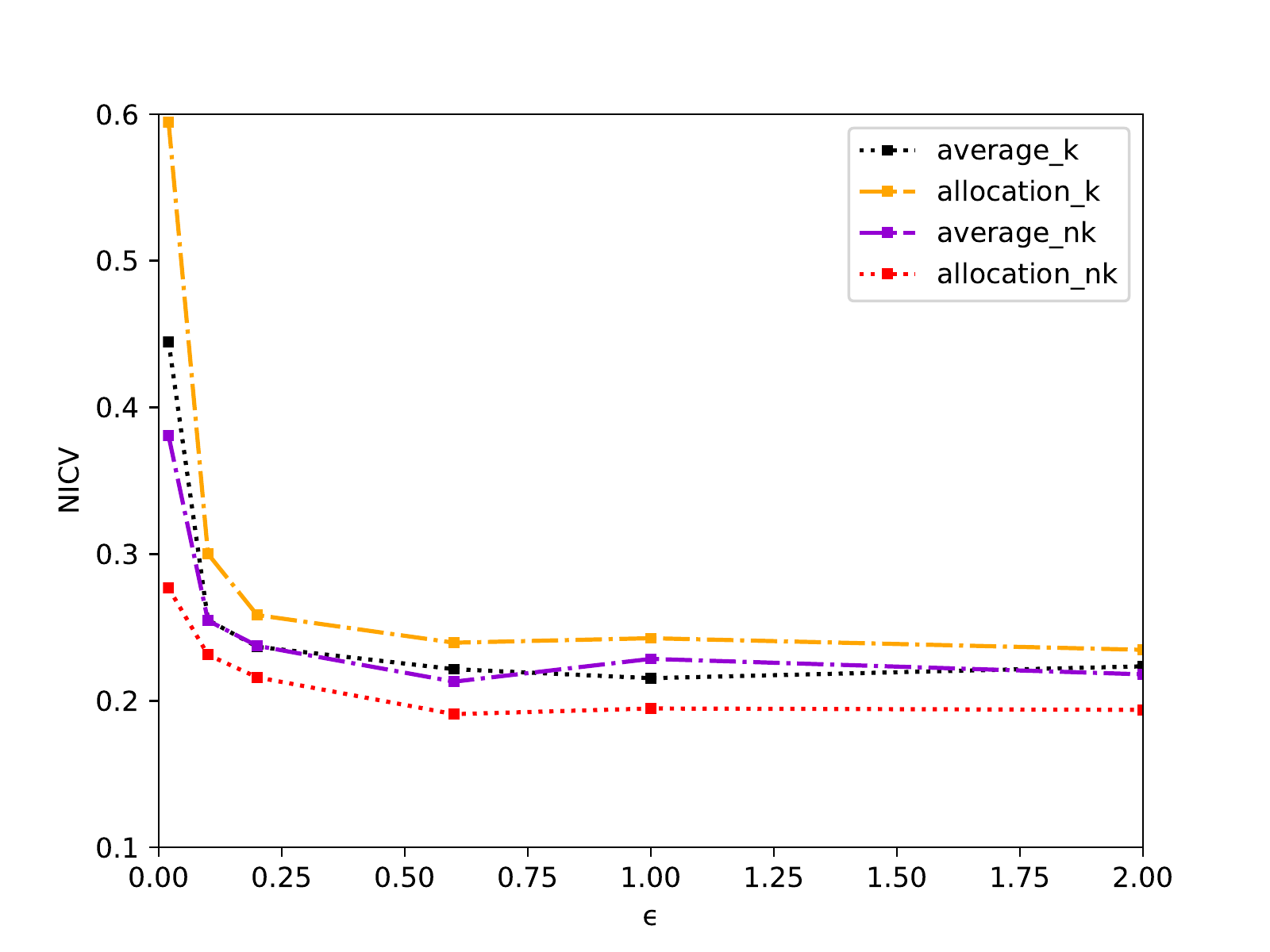}
\caption{Adult ($k=5$)}\label{fig:Adult-e}
\end{minipage}

\begin{minipage}[b]{0.5\linewidth}
\centering
\includegraphics[width=1\textwidth]{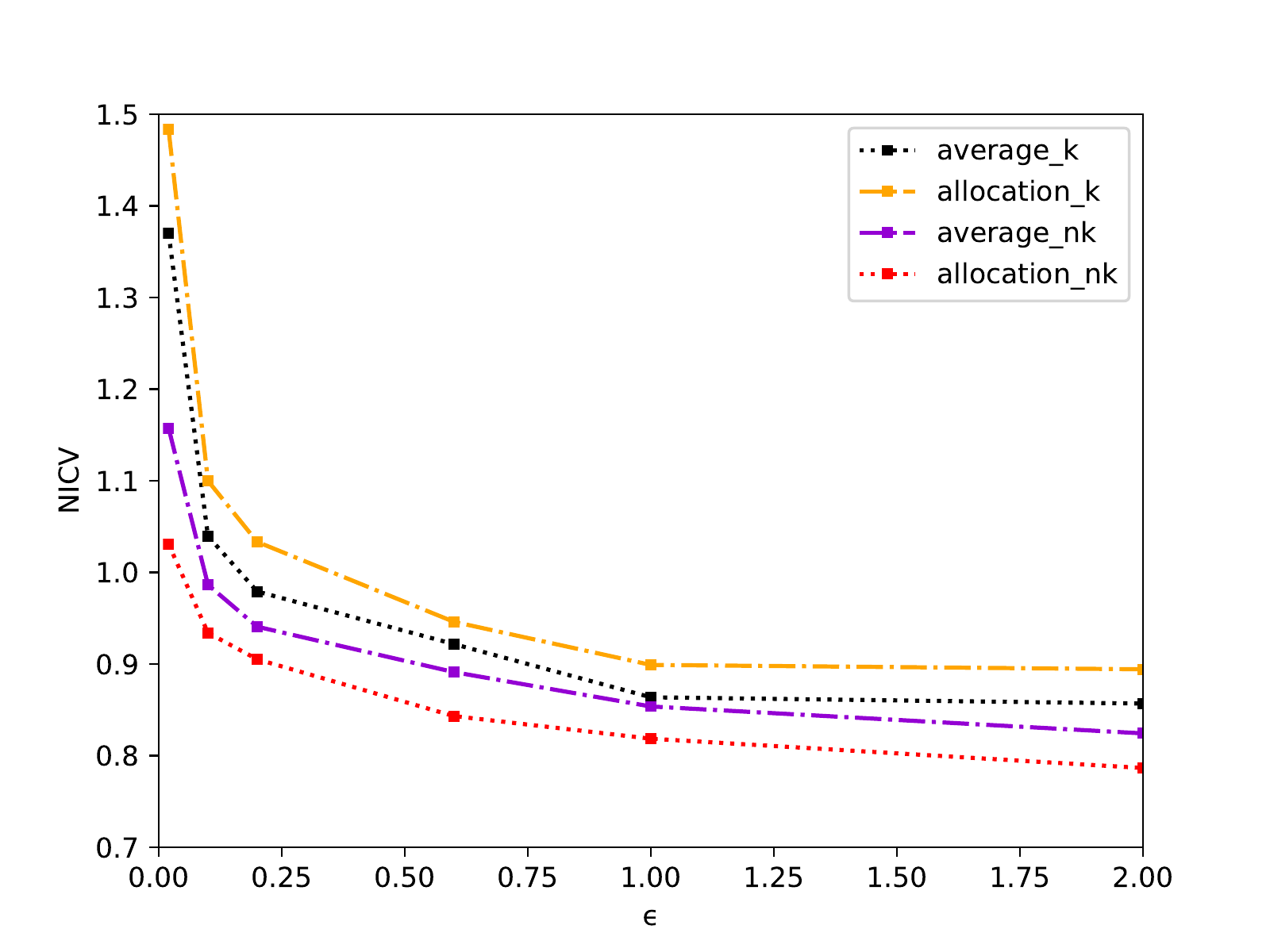}
\caption{Electrical ($k=5$)}\label{fig:Electrical-e}
\end{minipage}
\begin{minipage}[b]{0.5\linewidth}
\centering
\includegraphics[width=1\textwidth]{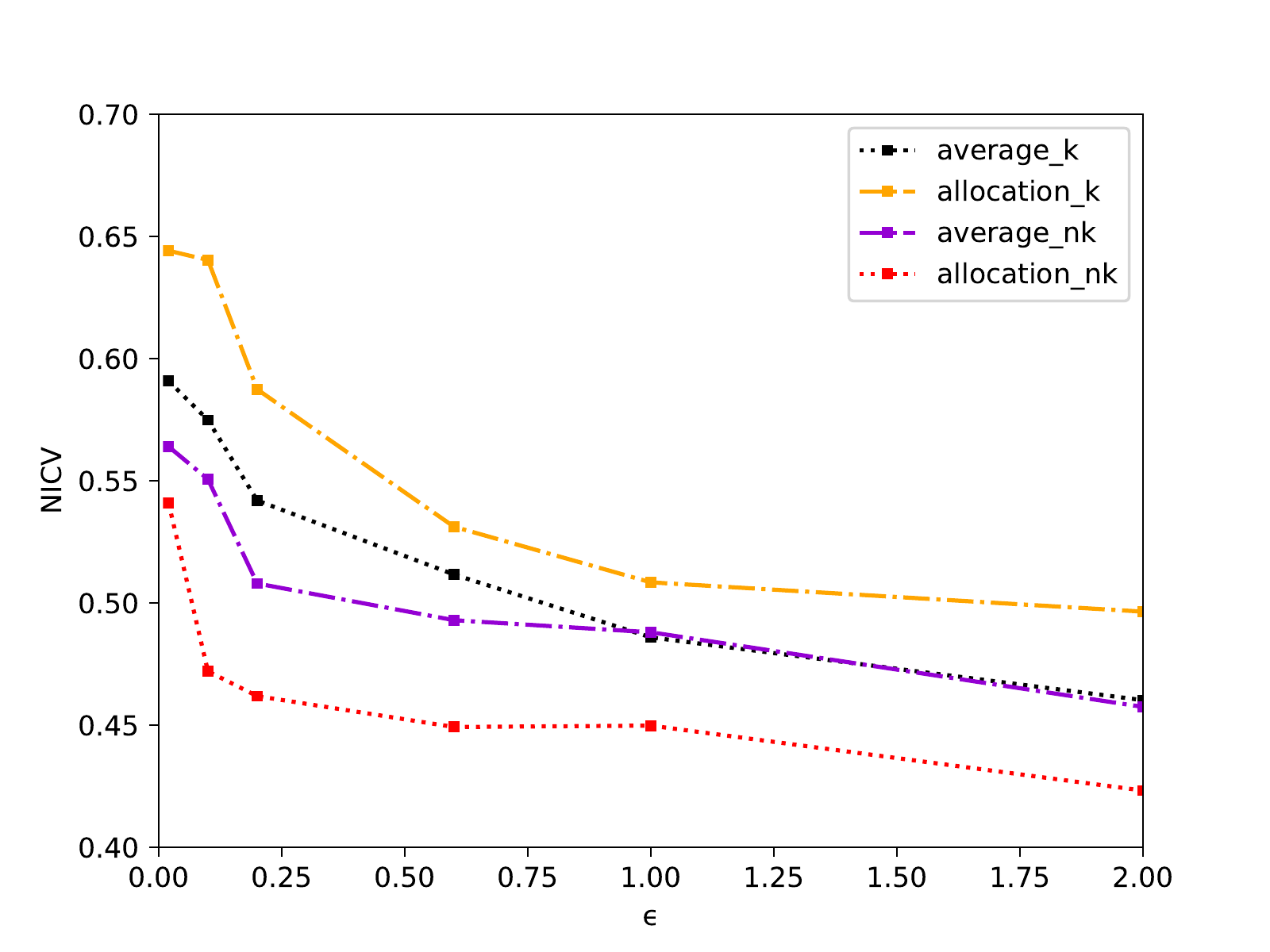}
\caption{Credit-card ($k=5$)}\label{fig:Credit-e}
\end{minipage}

\end{figure*}

\textbf{(2) Performance Comparison in term of $k$}

Figures \ref{fig:Blood-k} to \ref{fig:Credit-k} show the influence of different $k$ values on the clustering results for different datasets when the $\epsilon$ value is fixed. From these figures, we make observations as follows. 1) The result of performance comparisons between these four algorithms is similar to that of Figures \ref{fig:Blood-e} to \ref{fig:Credit-e}, and the ranking of the algorithms in decreasing order of the performance is $allocation\_nk$, $average\_nk$, $average\_k$, and $allocation\_k$. 2) For each figure, there is a point of certain $k$ value that has the most effect in utility improvement, which seems to show that this $k$ is the most appropriate number of clusters for the corresponding dataset.

\begin{figure*}[!htb]

\begin{minipage}[b]{0.5\linewidth}
\centering
\includegraphics[width=1\textwidth]{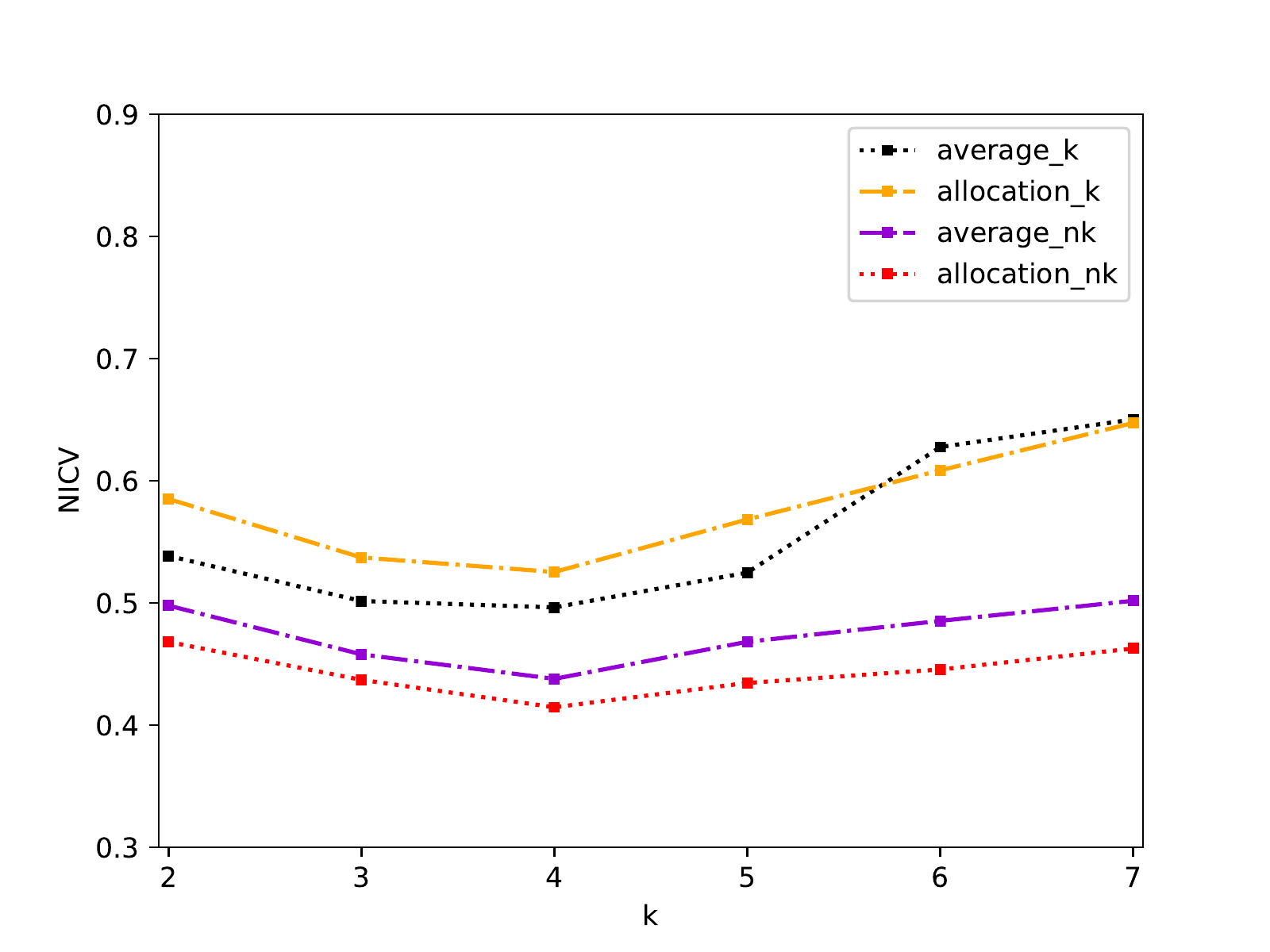}
\caption{Blood ($\epsilon=0.6$)}\label{fig:Blood-k}
\end{minipage}
\begin{minipage}[b]{0.5\linewidth}
\centering
\includegraphics[width=1\textwidth]{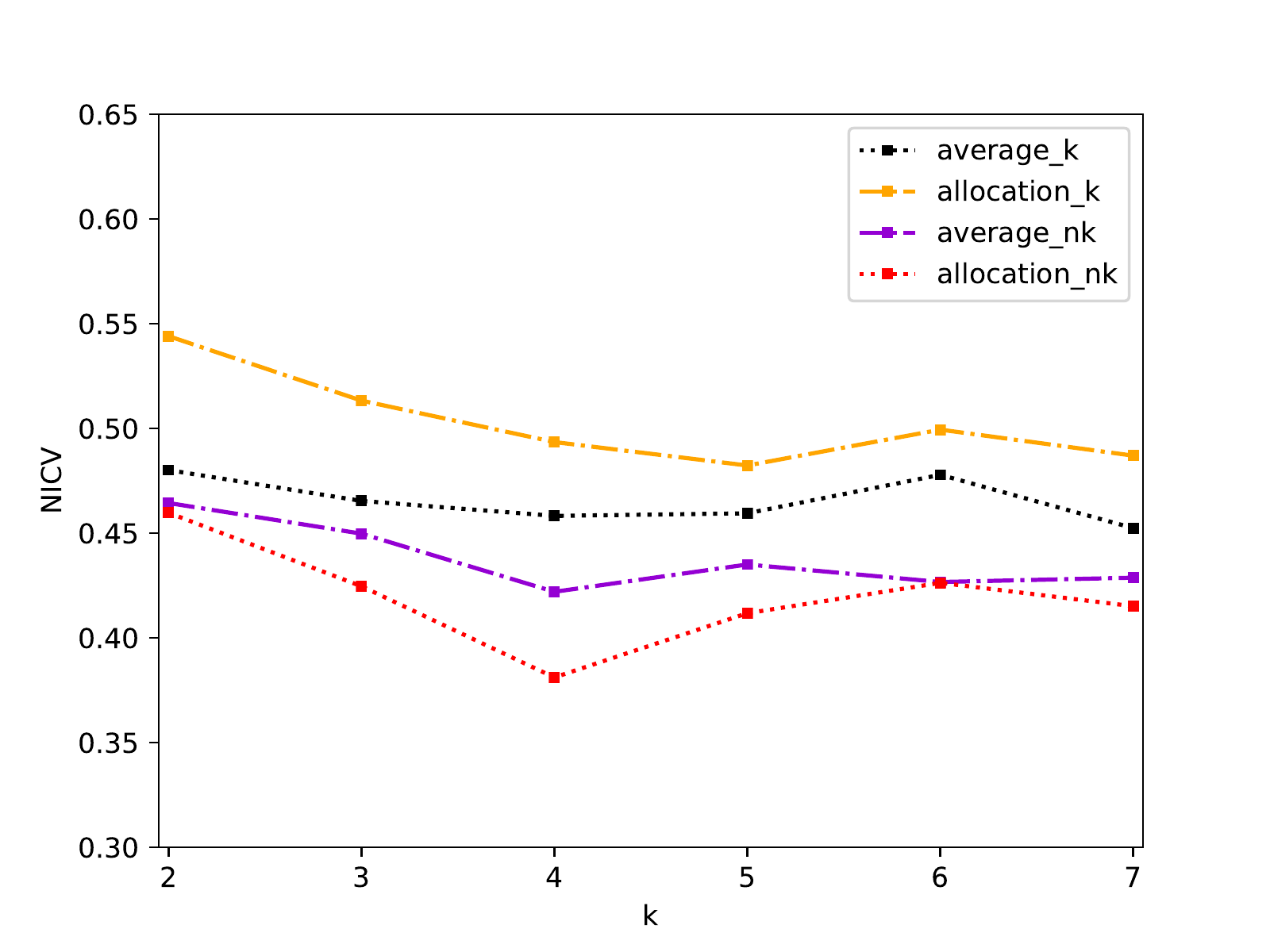}
\caption{Tripadvisor-review ($\epsilon=0.6$)}\label{fig:Tripadvisor-k}
\end{minipage}

\begin{minipage}[b]{0.5\linewidth}
\centering
\includegraphics[width=1\textwidth]{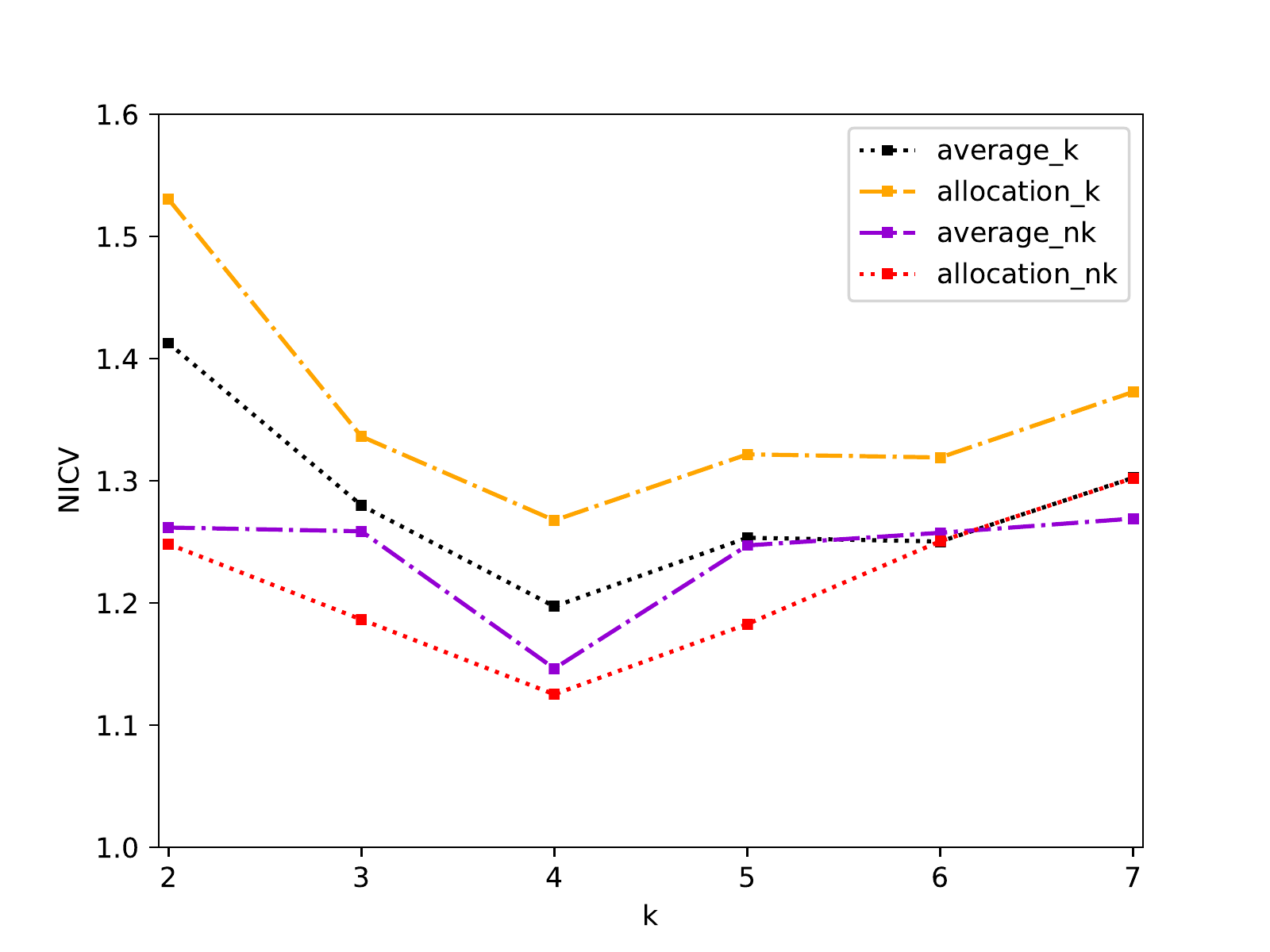}
\caption{Travel-review ($\epsilon=0.2$)}\label{fig:Travel-k}
\end{minipage}
\begin{minipage}[b]{0.5\linewidth}
\centering
\includegraphics[width=1\textwidth]{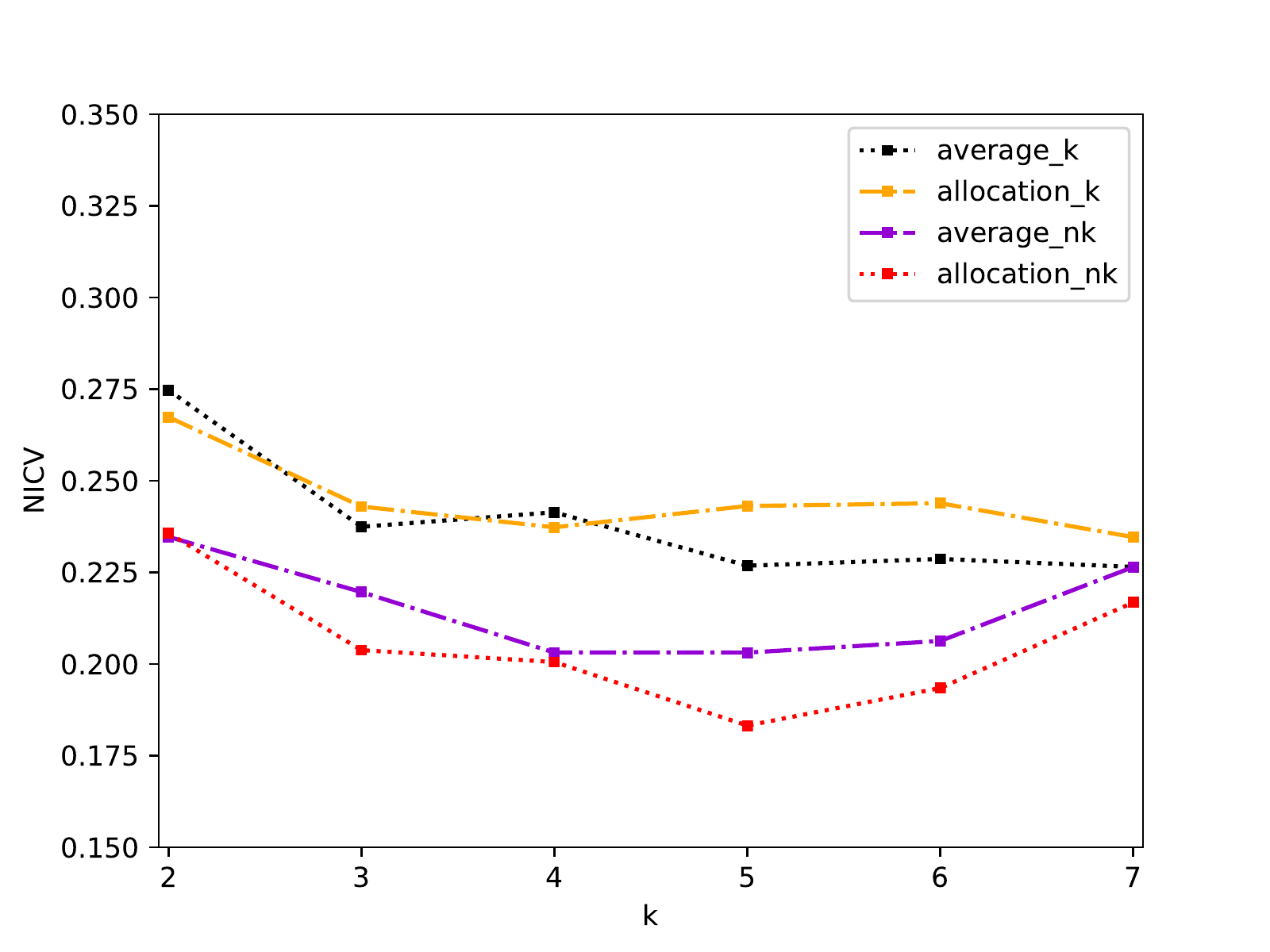}
\caption{Adult ($\epsilon=0.2$)}\label{fig:Adult-k}
\end{minipage}

\begin{minipage}[b]{0.5\linewidth}
\centering
\includegraphics[width=1\textwidth]{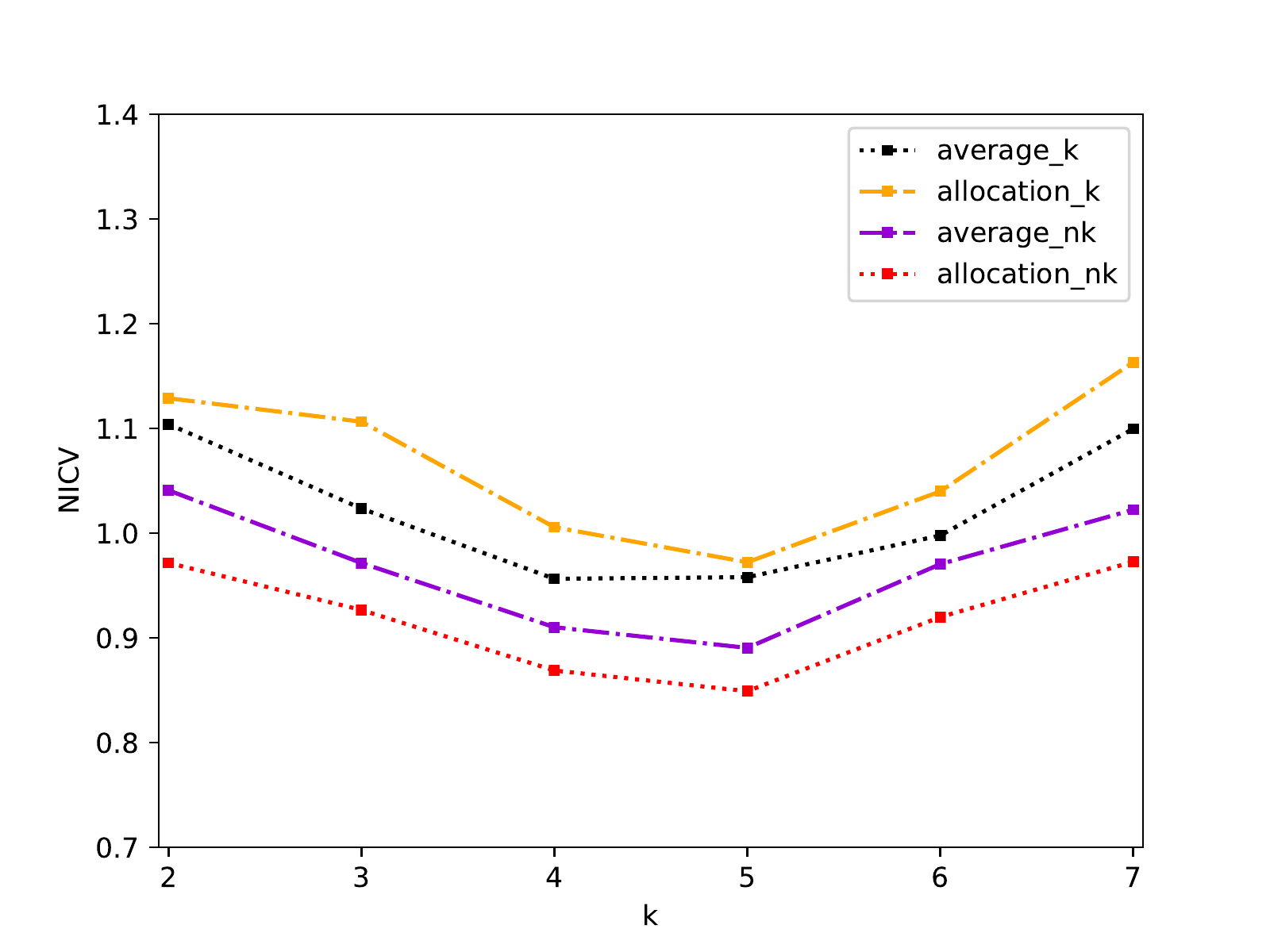}
\caption{Electrical ($\epsilon=0.2$)}\label{fig:Electrical-k}
\end{minipage}
\begin{minipage}[b]{0.5\linewidth}
\centering
\includegraphics[width=1\textwidth]{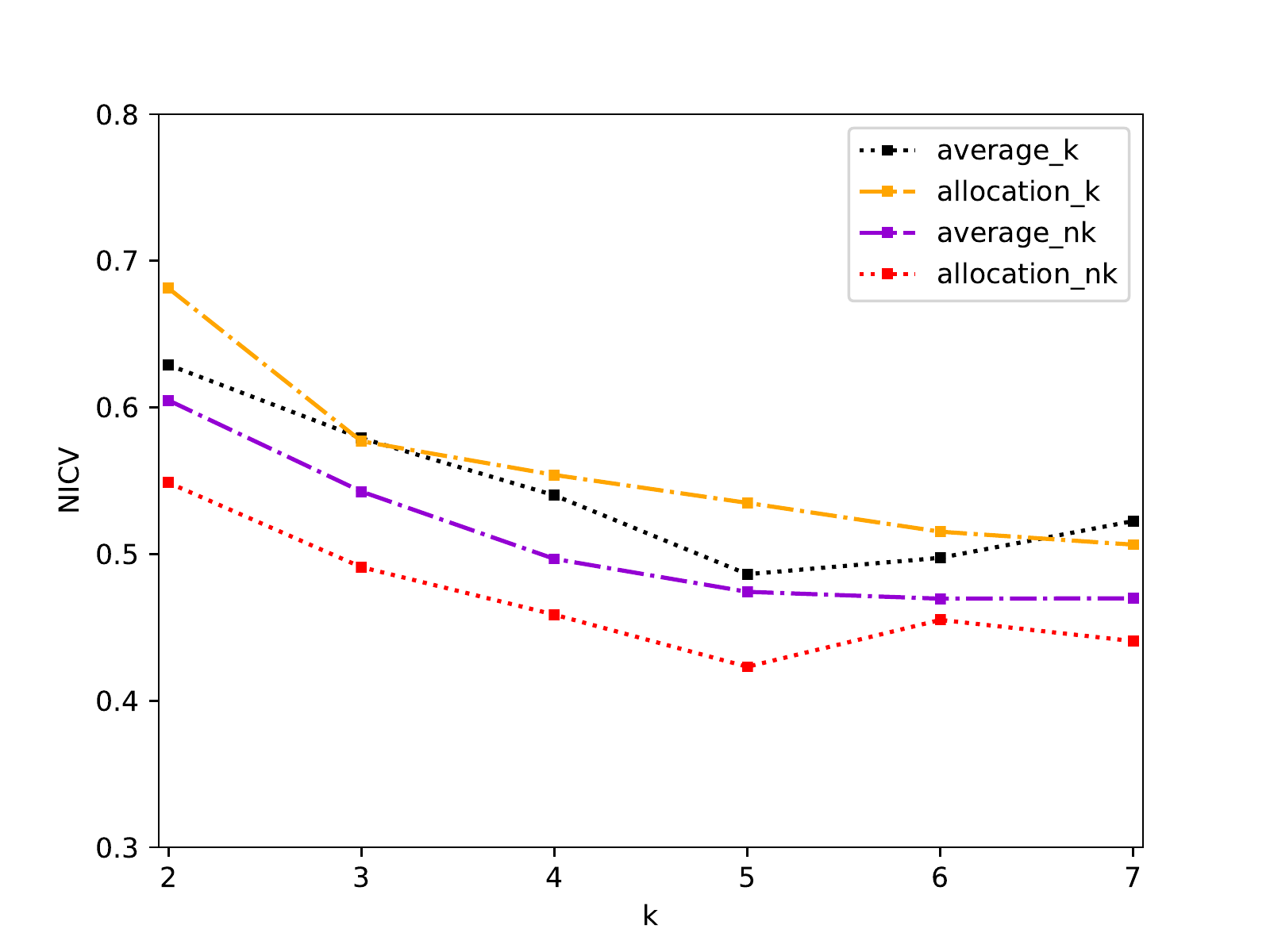}
\caption{Credit-card ($\epsilon=0.2$)}\label{fig:Credit-k}
\end{minipage}
\end{figure*}

Based on the experimental analysis above, we can conclude that the proposed algorithm combining both cluster merging and adaptive privacy budget allocation improves the clustering utility significantly, and it is superior to the state-of-the-art algorithms.

We can explain our experimental results with the stability theory by David et al. \cite{ben2006sober}. According to the definition, given a clustering algorithm, the stability is mainly determined by the probability distribution of the data. Specifically, if the data distribution is of some non-trivial symmetry structure, the algorithm would be unstable; otherwise, it would be stable. When adding differential privacy to the clustering algorithm, the data distribution is flattened. This makes the distribution is more symmetry, and the clustering becomes more unstable. In the extreme, when $\epsilon = 0$, the distribution is uniformly symmetry, and the clustering is completely unstable. Thus, adding differential privacy deteriorates the clustering. When merging the nearest clusters, the clustering algorithm makes the clusters more asymmetry by canceling out noises added, and the clustering becomes more stable. So, we can see that merging clusters improves the clustering utility. Furthermore, in the context of merging clusters, the privacy budget allocation works well. The underlying reason may be that the cluster merging allows more unstability in the first several iterations, while allows less in the last ones.

\section{CONCLUSION}\label{sec:conclusion}
In this paper, we have proposed a differentially private $k$-means clustering algorithm based on cluster merging. This algorithm improves the utility of $k$-means clustering by first partitioning the data into more clusters than required, and then merging the clusters into required number of clusters. It is shown that this cluster merging improves the clustering utility, and when combined with privacy budget allocation, it can further improves the utility. Extensive experiments show that our algorithm outperforms the state-of-the-art algorithms significantly. Besides, we only consider numerical data, without taking into account the non-numerical and mixed data. The future work is to design differentially private $k$-means clustering algorithms for more types of data.

%
%
%
%

\section*{References}

\bibliography{neuro}

\end{document}